\documentclass{llncs}
\usepackage{amssymb,amsmath,latexsym,graphicx,fancyhdr,url,subfigure,xspace}

\usepackage{algorithm}
\usepackage[noend]{algorithmic}

\newcommand{\etal}{\emph{et~al.}}
\newcommand{\fd}{Fr\'echet distance}
\newcommand{\fm}{Fr\'echet matching}

\renewcommand{\paragraph}[1]{\smallskip\noindent{\bf\sffamily #1.}}

\let\doendproof\endproof
\renewcommand\endproof{~\hfill\qed\doendproof}

\graphicspath{{img/}}

\newcommand{\shrinkfiguregap}{\vspace{-0.4\baselineskip}}

\title{Locally Correct Fr\'echet Matchings\thanks{M. Buchin is supported by the Netherlands Organisation for Scientific Research (NWO) under project no.~612.001.106. W. Meulemans and B. Speckmann are supported by the Netherlands Organisation for Scientific Research (NWO) under project no.~639.022.707. A preliminary version of this paper will appear in Proc.\ 20th European Symposium on Algorithms (ESA 2012).}} 

\author{Kevin Buchin \and Maike Buchin \and Wouter Meulemans \and Bettina Speckmann
\institute{Dep. of Mathematics and Computer Science, TU Eindhoven, The Netherlands.\\ \email{k.a.buchin@tue.nl} \qquad \email{m.e.buchin@tue.nl} \qquad \email{w.meulemans@tue.nl} \qquad \email{speckman@win.tue.nl}}
}

\begin{document}
\mainmatter

\maketitle
\thispagestyle{plain}

\begin{abstract}
The \fd\ is a metric to compare two curves, which is based on monotonous matchings between these curves. We call a matching that results in the \fd\ a \fm. There are often many different \fm s and not all of these capture the similarity between the curves well. We propose to restrict the set of \fm s to ``natural'' matchings and to this end introduce \emph{locally correct} \fm s. We prove that at least one such matching exists for two polygonal curves and give an $O(N^3 \log N)$ algorithm to compute it, where $N$ is the total number of edges in both curves. We also present an $O(N^2)$ algorithm to compute a locally correct discrete \fm.
\end{abstract}

\section{Introduction}\label{sec:introduction}

Many problems ask for the comparison of two curves. Consequently, several distance measures have been proposed for the similarity of two curves $P$ and $Q$, for example, the Hausdorff and the \fd. Such a distance measure simply returns a number indicating the (dis)similarity.
However, the Hausdorff and the \fd\ are both based on matchings of the points on the curves.
The distance returned is the maximum distance between any two matched points.
The \fd\ uses \emph{monotonous matchings} (and limits of these): if point $p$ on $P$ and $q$ on $Q$ are matched, then any point on $P$ after $p$ must be matched to $q$ or a point on $Q$ after $q$.
The \emph{\fd}\ is the maximal distance between two matched points minimized over all monotonous matchings of the curves.
Restricting to monotonous matchings of only the vertices results in the \emph{discrete \fd}.
We call a matching resulting in the (discrete) \fd\ a \emph{(discrete) \fm}. See Section~\ref{sec:prelims} for more details.

There are often many different \fm s for two curves.
However, as the \fd\ is determined only by the maximal distance, not all of these matchings capture the similarity between the curves well (see Fig.~\ref{fig:badmatch}).
There are applications that directly use a matching, for example, to map a GPS track to a street network~\cite{Wenk2006} or to morph between the curves~\cite{Efrat2002}.
In such situations
a ``good'' matching is important.
We believe that many applications of the (discrete) \fd, such as protein alignment~\cite{Wylie2012} and detecting patterns in movement data~\cite{bbgll-dcpcs-11}, would profit from good \fm s.

\begin{figure}[t]
  \centering
  \includegraphics{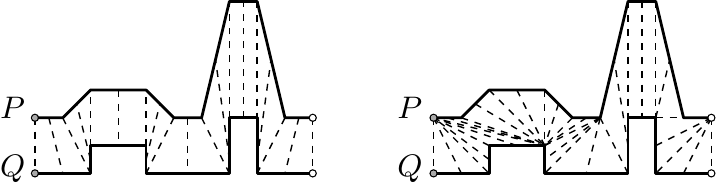}
  \shrinkfiguregap
  \caption{Two Fr\'echet matchings for curves $P$ and $Q$.}
  \label{fig:badmatch}
\end{figure}

\begin{figure}[b]
  \centering
  \includegraphics{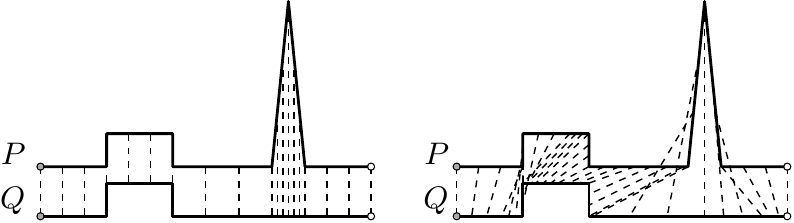}
  \shrinkfiguregap
  \caption{Two Fr\'echet matchings. Right: the result of speed limits is not locally correct.}
  \label{fig:speedcounter}
\end{figure}

\paragraph{Results}
We restrict the set of \fm s to ``natural'' matchings by introducing \emph{locally correct} \fm s:
matchings that for any two matched subcurves are again a \fm\ on these subcurves.
In Section~\ref{sec:lcm} we prove that there exists such a locally correct \fm\ for any two polygonal curves.
Based on this proof we describe in Section~\ref{sec:algorithm} an $O(N^3 \log N)$ algorithm to compute such a matching, where $N$ is the total number of edges in both curves.
We consider the discrete \fd\ in Section~\ref{sec:discrete} and give an $O(N^2)$ algorithm to compute locally correct matchings under this metric.

\paragraph{Related work}
The first algorithm to compute the \fd\ was given by Alt and Godau~\cite{Alt1995}.
They also consider a non-monotone \fd\ and their algorithm for this variant results in a locally correct non-monotone matching (see Remark 3.5 in~\cite{hr-fdre-12}).
Eiter and Mannila gave the first algorithm to compute the discrete \fd~\cite{em-cdfd-94}.
Since then, the \fd\ has received significant attention. Here we focus on approaches that restrict the allowed matchings.
Efrat~\etal~\cite{Efrat2002} introduced Fr\'echet-like metrics, the geodesic width and link width, to restrict to matchings suitable for curve morphing.
Their method is suitable only for non-intersecting polylines.
Moreover, geodesic width and link width do not resolve the problem illustrated in Fig.~\ref{fig:badmatch}: both matchings also have minimal geodesic width and minimal link width.
Maheshwari~\etal~\cite{Maheshwari2011} studied a restriction by ``speed limits'', which may exclude all \fm s and may cause undesirable effects near ``outliers'' (see Fig.~\ref{fig:speedcounter}).
Buchin~\etal~\cite{Buchin2010} describe a framework for restricting \fm s, which they illustrate by restricting slope and path length.
The former corresponds to speed limits.
We briefly discuss the latter at the end of Section~\ref{sec:algorithm}.

\section{Preliminaries}\label{sec:prelims}

{\sffamily\bfseries Curves.} Let $P$ be a polygonal curve with $m$ edges, defined by vertices $p_0, \ldots, p_m$. We treat a curve as a continuous map $P : [0,m] \rightarrow \mathbb{R}^d$. In this map, $P(i)$ equals $p_i$ for integer $i$. Furthermore, $P(i+\lambda)$ is a parameterization of
   the $(i+1)$st edge, that is,  $P(i+\lambda) = (1-\lambda) \cdot p_i + \lambda \cdot p_{i+1}$, for integer $i$ and $0 < \lambda < 1$.
As a reparametrization $\sigma : [0,1] \rightarrow [0,m]$ of a curve~$P$, we allow any continuous, non-decreasing function such that $\sigma(0) = 0$ and $\sigma(1) = m$. We denote by $P_\sigma(t)$ the actual location according to reparametrization $\sigma$: $P_\sigma(t) = P(\sigma(t))$. By $P_\sigma[a,b]$ we denote the subcurve of $P$ in between $P_\sigma(a)$ and $P_\sigma(b)$. In the following we are always given two polygonal curves $P$ and $Q$, where $Q$ is defined by its vertices $q_0, \ldots, q_n$ and is reparametrized by $\theta  : [0,1] \rightarrow [0,n]$. The reparametrized curve is denoted by $Q_\theta$.

\paragraph{\fm s}
We are given two polygonal curves $P$ and $Q$ with $m$ and $n$ edges. A (monotonous) \emph{matching} $\mu$ between $P$ and $Q$ is a pair of reparametrizations $(\sigma,\theta)$, such that $P_\sigma(t)$ matches to $Q_\theta(t)$.  The Euclidean distance between two matched points is denoted by $d_\mu(t) = |P_\sigma(t) - Q_\theta(t)|$. The maximum distance over a range is denoted by $d_\mu[a,b] = \max_{a \leq t \leq b} d_\mu(t)$. The \emph{\fd}\ between two curves is defined as $\delta_\text{F}(P,Q) = \inf_{\mu} d_\mu[0,1]$.
A \emph{\fm}\ is a matching $\mu$ that realizes the \fd: $d_\mu[0,1] = \delta_\text{F}(P,Q)$ holds.

\paragraph{Free space diagrams}
Alt and Godau~\cite{Alt1995} describe an algorithm to compute the \fd\ based on the decision variant (that is, solving $\delta_\text{F}(P,Q) \leq \varepsilon$ for some given $\varepsilon$).
Their algorithm uses a \emph{free space diagram}, a two-dimensional diagram on the range $[0, m] \times [0, n]$.
Every point $(x,y)$ in this diagram is either ``free'' (white) or not (indicating whether $|P(x) - Q(y)| \leq \varepsilon$).
The diagram has $m$ columns and $n$ rows; every cell $(c,r)$ ($1 \leq c \leq m$ and $1 \leq r \leq n$) corresponds to the edges $p_{c-1}p_{c}$ and $q_{r-1}q_{r}$.
To compute the \fd, one finds the smallest $\varepsilon$ such that there exists an x- and y-monotone path from point $(0,0)$ to $(m, n)$ in free space.
\begin{figure}[b]
  \centering
  \includegraphics{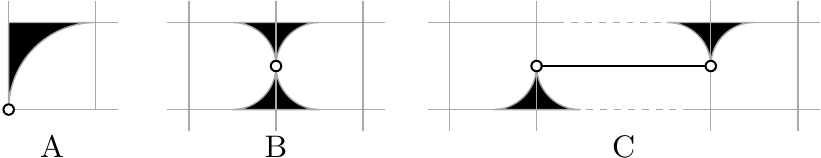}
  \shrinkfiguregap
  \caption{Three event types. (A) Endpoints come within range of each other. (B) Passage opens on cell boundary. (C) Passage opens in row (or column). We scale every row and column in the diagram to correspond to the (relative) length of the actual edge of the curve instead of using unit squares for cells.}
  \label{fig:events}
\end{figure}
For this, only certain \emph{critical values} for the distance have to be checked. Imagine continuously increasing the distance $\varepsilon$ starting at $\varepsilon = 0$. At so-called \emph{critical events}, which are illustrated in Fig.~\ref{fig:events}, passages open in the free space. 
The critical values are the distances corresponding to these events.

\section{Locally correct \fm s}\label{sec:lcm}

We introduce \emph{locally correct} \fm s, for which the matching between any two matched subcurves is a \fm.

\begin{definition}[Local correctness]\label{def:locallycorrect}
Given two polygonal curves $P$ and $Q$, a matching $\mu = (\sigma,\theta)$ is \emph{locally correct} if for all $a,b$ with $0 \leq a \leq b \leq 1$
\[ d_\mu[a,b] = \delta_\text{F}(P_\sigma[a,b], Q_\theta[a,b]). \]
\end{definition}
Note that not every \fm\ is locally correct. See for example Fig.~\ref{fig:speedcounter}.
The question arises whether a locally correct matching always exists and if so, how to compute it.
We resolve the first question in the following theorem.

\begin{theorem}\label{thm:main}
For any two polygonal curves $P$ and $Q$, there exists a locally correct \fm.
\end{theorem}

\paragraph{Existence}
We prove Theorem~\ref{thm:main} by induction on the number of edges in the curves.
First, we present the lemmata for the two base cases: one of the two curves is a point, and both curves are line segments. In the following, $n$ and $m$ again denote the number of edges of $P$ and $Q$, respectively.

\begin{lemma}\label{lem:pointcurve}
For two polygonal curves $P$ and $Q$ with $m = 0$, a locally correct matching is $(\sigma, \theta)$, where $\sigma(t) = 0$ and $\theta(t) = t \cdot n$.
\end{lemma}
\begin{proof}
Since $m = 0$, $P$ is just a single point, $p_0$. The \fd\ between a point and a curve is the maximal distance between the point and any point on the curve: $\delta_\text{F}(p_0, Q_\theta[a,b]) = d_\mu[a,b]$. This implies that the matching $\mu$ is locally correct.
\end{proof}

\begin{lemma}\label{lem:linesegments}
For two polygonal curves $P$ and $Q$ with $m = n = 1$, a locally correct matching is $(\sigma, \theta)$, where $\sigma(t) = \theta(t) = t$.
\end{lemma}
\begin{proof}
The free space diagram of $P$ and $Q$ is a single cell and thus the free space is a convex area for any value of $\varepsilon$. Since $\mu = (\sigma, \theta)$ is linear, we have that $d_\mu[a,b] = \max\left\{ d_\mu(a), d_\mu(b) \right\}$: if there would be a $t$ with $a < t < b$ such that $d_\mu(t) > \max\left\{ d_\mu(a), d_\mu(b) \right\}$, then the free space at $\varepsilon = \max\left\{ d_\mu(a), d_\mu(b) \right\}$  would not be convex.
Since $d_\mu[a,b] = \max\left\{ d_\mu(a), d_\mu(b) \right\} \leq \delta_\text{F}(P_\sigma[a,b],Q_\theta[a,b]) \leq d_\mu[a,b]$, we conclude that $\mu$ is locally correct.
\end{proof}

\begin{figure}[t]
  \centering
  \includegraphics{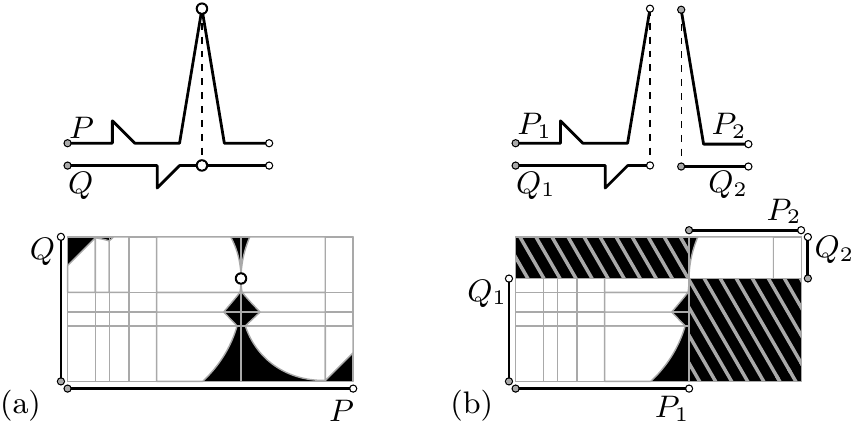}
  \shrinkfiguregap
  \caption{(a) Curves with the free space diagram for $\varepsilon = \delta_\text{F}(P,Q)$ and the realizing event. (b) The event splits each curve into two subcurves. The hatched areas indicate parts that disappear after the split.}
  \label{fig:split}
\end{figure}
For induction, we split the two curves based on events (see Fig.~\ref{fig:split}).
Since each split must reduce the problem size, we ignore any events on the left or bottom boundary of cell $(1,1)$ or on the right or top boundary of cell $(m,n)$. This excludes both events of type A. A free space diagram is \emph{connected} at value $\varepsilon$, if a monotonous path exists from the boundary of cell $(1,1)$ to the boundary of cell $(m,n)$. A \emph{realizing event} is a critical event at the minimal value $\varepsilon$ such that the corresponding free space diagram is connected.

Let $\mathcal{E}$ denote the set of concurrent realizing events for two curves.
A \emph{realizing set} $E_\text{r}$ is a subset of $\mathcal{E}$ such that the free space admits a monotonous path from cell $(1,1)$ to cell $(m,n)$ without using an event in $\mathcal{E} \backslash E_\text{r}$.
Note that a realizing set cannot be empty.
When $\mathcal{E}$ contains more than one realizing event, some may be ``insignificant'': they are never required to actually make a path in the free space diagram.
A realizing set is \emph{minimal} if it does not contain a strict subset that is a realizing set. Such a minimal realizing set contains only ``significant'' events.

\begin{lemma}\label{lem:realizingset}
For two polygonal curves $P$ and $Q$ with $m > 1$ and $n \geq 1$, there exists a minimal realizing set.
\end{lemma}
\begin{proof}
Let $\mathcal{E}$ denote the non-empty set of concurrent events at the minimal critical value. By definition, the empty set cannot be a realizing set and $\mathcal{E}$ is a realizing set. Hence, $\mathcal{E}$ contains a minimal realizing set.
\end{proof}
The following lemma directly implies that a locally correct \fm\ always exists.
Informally, it states that curves have a locally correct matching that is ``closer''  (except in cell $(1,1)$ or $(m,n)$) than the distance of their realizing set. Further, this matching is linear inside every cell.
In the remainder, we use realizing set to indicate a minimal realizing set, unless indicated otherwise.
\begin{lemma}\label{lem:specmatching}
If the free space diagram of two polygonal curves $P$ and $Q$ is connected at value $\varepsilon$, then there exists a locally correct \fm\ $\mu = (\sigma,\theta)$ such that $d_\mu(t) \leq \varepsilon$ for all $t$ with $\sigma(t) \geq 1$ or $ \theta(t) \geq 1$, and $\sigma(t) \leq m-1$ or $\theta(t) \leq n -1$. Furthermore, $\mu$ is linear in every cell.
\end{lemma}
\begin{proof}
We prove this by induction on $m + n$. The base cases ($m = 0$, $n = 0$, and $m = n = 1$) follow from Lemma~\ref{lem:pointcurve} and Lemma~\ref{lem:linesegments}.

For induction, we assume that $m \geq 1$, $n \geq 1$, and $m + n > 2$.
By Lemma~\ref{lem:realizingset}, a realizing set $E_\text{r}$ exists for $P$ and $Q$, say at value $\varepsilon_\text{r}$.
The set contains realizing events $e_1, \ldots, e_k$ ($k \geq 1$), numbered in lexicographic order.
By definition, $\varepsilon_\text{r} \leq \varepsilon$ holds.
Suppose that $E_\text{r}$ splits curve $P$ into $P_1, \ldots, P_{k+1}$ and curve $Q$ into $Q_1, \ldots, Q_{k+1}$, where $P_i$ has $m_i$ edges, $Q_i$ has $n_i$ edges.
By definition of a realizing event, none of the events in $E_\text{r}$ occur on the right or top boundary of cell $(m,n)$.
Hence, for any $i$ ($1 \leq i \leq k+1$), it holds that $m_i \leq m$, $n_i \leq n$, and $m_i < m$ or $n_i < n$.
Since a path exists in the free space diagram at $\varepsilon_\text{r}$ through all events in $E_\text{r}$, the induction hypothesis implies that, for any $i$ ($1 \leq i \leq k+1$), a locally correct matching $\mu_i = (\sigma_i,\theta_i)$ exists for $P_i$ and $Q_i$ such that $\mu_i$ is linear in every cell and $d_{\mu_i}(t) \leq \varepsilon_\text{r}$ for all $t$ with $\sigma_i(t) \geq 1$ or $\theta_i(t) \geq 1$, and $\sigma_i(t) \leq m_i-1$ or $\theta_i(t) \leq n_i -1$.
Combining these matchings with the events in $E_\text{r}$ yields a matching $\mu = (\sigma,\theta)$ for $(P,Q)$.

As we argue below, this matching is locally correct and satisfies the additional properties.
The matching of an event corresponds to a single point (type B) or a horizontal or vertical line (type C).
By induction, $\mu_i$ is linear in every cell.
Since all events occur on cell boundaries, the cells of the matchings and events are disjoint.
Therefore, the matching $\mu$ is also linear inside every cell.

For $i < k+1$, $d_{\mu_i}$ is at most $\varepsilon_\text{r}$ at the point where $\mu_i$ enters cell $(m_i,n_i)$ in the free space diagram of $P_i$ and $Q_i$. We also know that $d_{\mu_i}$ equals $\varepsilon_\text{r}$ at the top right corner of cell $(m_i,n_i)$. Since $\mu_i$ is linear inside the cell, $d_{\mu_i}(t) \leq \varepsilon_\text{r}$ also holds for $t$ with $\sigma_i(t) > m_i-1$ and $\theta_i(t) > n_i -1$. Analogously, for $i > 0$, $d_{\mu_i}(t)$ is at most $\varepsilon_\text{r}$ for $t$ with $\sigma_i(t) < 1$ and $\theta_i(t) < 1$. Hence, $d_\mu(t) \leq \varepsilon_\text{r} \leq \varepsilon$ holds for $t$ with $\sigma(t) \geq 1$ or $\theta(t) \geq 1$, and $\sigma(t) \leq m-1$ or $\theta(t) \leq n -1$.

To show that $\mu$ is locally correct, suppose for contradiction that values $a,b$ exist such that $\delta_\text{F}(P_\sigma[a,b], Q_\theta[a,b]) < d_\mu[a,b]$. If $a,b$ are in between two consecutive events, we know that the submatching corresponds to one of the matchings $\mu_i$. Since these are locally correct, $\delta_\text{F}(P_\sigma[a,b], Q_\theta[a,b]) = d_\mu[a,b]$ must hold.

Hence, suppose that $a$ and $b$ are separated by at least one event of $E_\text{r}$.
There are two possibilities: either $d_\mu[a,b] = \varepsilon_\text{r}$ or $d_\mu[a,b] > \varepsilon_\text{r}$.
$d_\mu[a,b] < \varepsilon_\text{r}$ cannot hold, since $d_\mu[a,b]$ includes a realizing event.
First, assume $d_\mu[a,b] = \varepsilon_\text{r}$ holds.
If $\delta_\text{F}(P_\sigma[a,b], Q_\theta[a,b]) < \varepsilon_\text{r}$ holds, then a matching exists that does not use the events between $a$ and $b$ and has a lower maximum.
Hence, the free space connects point $(\sigma(a),\theta(a))$ with point $(\sigma(b),\theta(b))$ at a lower value than $\varepsilon_\text{r}$. This implies that all events between $a$ and $b$ can be omitted, contradicting that $E_\text{r}$ is a minimal realizing set.

Now, assume $d_\mu[a,b] > \varepsilon_\text{r}$. Let $t'$ denote the highest $t$ for which $\sigma(t) \leq 1$ and $\theta(t) \leq 1$ holds, that is, the point at which the matching exits cell $(1,1)$.
Similarly, let $t''$ denote the lowest $t$ for which $\sigma(t) \geq m-1$ and $\theta(t) \geq n-1$ holds. Since $d_\mu(t) \leq \varepsilon_\text{r}$ holds for any $t' \leq t \leq t''$, $d_\mu(t) > \varepsilon_\text{r}$ can hold only for $t < t'$ or $t > t''$. Suppose that $d_\mu(a) > \varepsilon_\text{r}$ holds. Then $a < t'$ holds and $\mu$ is linear between $a$ and $t'$. Therefore, $d_\mu(a) > d_\mu(t)$ holds for any $t$ with $a < t < t'$.
Analogously, if $d_\mu(b) > \varepsilon_\text{r}$ holds, then $d_\mu(b) > d_\mu(t)$ holds for any $t$ with $t'' < t < b$ . Hence, $d_\mu[a,b] = \max \left\{ d_\mu(a), d_\mu(b) \right\}$ must hold.

This maximum is a lower bound on the \fd, contradicting the assumption that $d_\mu[a,b]$ is larger than the \fd. Matching $\mu$ is therefore locally correct.
\end{proof}

\section{Algorithm for locally correct \fm s}\label{sec:algorithm}

The existence proof directly results in a recursive algorithm, which is given by Algorithm~\ref{alg:findmatching}.
Fig.~\ref{fig:badmatch} (left), Fig.~\ref{fig:speedcounter} (left), Fig.~\ref{fig:examplematching}, Fig.~\ref{fig:vertexdependency}, and Fig.~\ref{fig:lengthcounter} (left) illustrate matchings computed with our algorithm. This section is devoted to proving the following theorem.
\begin{theorem}
Algorithm~\ref{alg:findmatching} computes a locally correct \fm\ of two polygonal curves $P$ and $Q$ with $m$ and $n$ edges in $O((m+n) m n \log (mn))$ time.
\end{theorem}

\begin{algorithm}
  \caption{$\texttt{ComputeLCFM}(P, Q)$}\label{alg:findmatching}
  \begin{algorithmic}[1]

  \REQUIRE $P$ and $Q$ are curves with $m$ and $n$ edges
  \ENSURE A locally correct \fm\ for $P$ and $Q$

  \IF{$m = 0$ \textbf{or} $n = 0$}
    \RETURN $(\sigma, \theta)$ where $\sigma(t) = t \cdot m$, $\theta(t) = t \cdot n$
  \ELSIF{$m = n = 1$}
    \RETURN $(\sigma, \theta)$ where $\sigma(t) = \theta(t) = t$
  \ELSE
    \STATE Find event $e_\text{r}$ of a minimal realizing set \label{algline:find_er}
    \STATE Split $P$ into $P_1$ and $P_2$ according to $e_\text{r}$
    \STATE Split $Q$ into $Q_1$ and $Q_2$ according to $e_\text{r}$
    \STATE $\mu_1 \rightarrow \texttt{ComputeLCFM}(P_1, Q_1)$
    \STATE $\mu_2 \rightarrow \texttt{ComputeLCFM}(P_2, Q_2)$
    \RETURN concatenation of $\mu_1$, $e_\text{r}$, and $\mu_2$
  \ENDIF
 \end{algorithmic}
\end{algorithm}

Using the notation of Alt and Godau~\cite{Alt1995}, $L^F_{i,j}$ denotes the interval of free space on the left boundary of cell $(i,j)$; $L^R_{i,j}$ denotes the subset of $L^F_{i,j}$ that is reachable from point $(0,0)$ of the free space diagram with a monotonous path in the free space. Analogously, $B^F_{i,j}$ and $B^R_{i,j}$ are defined for the bottom boundary.

With a slight modification to the decision algorithm, we can compute the minimal value of $\varepsilon$ such that a path is available from cell $(1,1)$ to cell $(m,n)$. This requires only two changes: $B^R_{1,2}$ should be initialized with $B^F_{1,2}$ and $L^R_{2,1}$ with $L^F_{2,1}$; the answer should be ``yes'' if and only if $B^R_{m,n}$ or $L^R_{m,n}$ is non-empty.

\begin{figure}[b]
  \centering
  \includegraphics{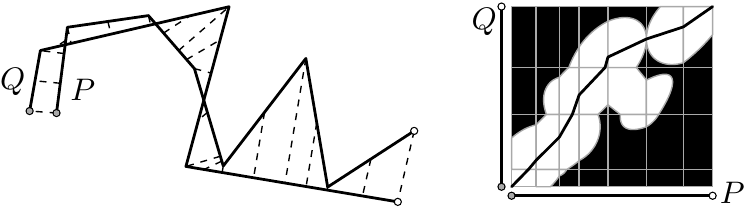}
  \shrinkfiguregap
  \caption{Locally correct matching produced by Algorithm~\ref{alg:findmatching}. Free space diagram drawn at $\varepsilon = \delta_\text{F}(P,Q)$.}
  \label{fig:examplematching}
\end{figure}

\paragraph{Realizing set}
By computing the \fd\ using the modified Alt and Godau algorithm, we obtain an ordered, potentially non-minimal realizing set $\mathcal{E} = \{ e_1, \ldots, e_l \}$.
The algorithm must find an event that is contained in a realizing set.
Let $E_{k}$ denote the first $k$ events of $\mathcal{E}$. For now we assume that the events in $\mathcal{E}$ end at different cell boundaries.
We use a binary search on $\mathcal{E}$ to find the $r$ such that $E_r$ contains a realizing set, but $E_{r-1}$ does not. This implies that event $e_r$ is contained in a realizing set and can be used to split the curves. Note that $r$ is unique due to monotonicity.
For correctness, the order of events in $\mathcal{E}$ must be consistent in different iterations, for example, by using a lexicographic order.
Set $E_r$ contains only realizing sets that use $e_r$.
Hence, $E_{r-1}$ contains a realizing set to connect cell $(1,1)$ to $e_r$ and $e_r$ to cell $(m,n)$.
Thus any event found in subsequent iterations is part of $E_{r-1}$ and of a realizing set with $e_r$.

To determine whether some $E_{k}$ contains a realizing set, we check whether cells $(1,1)$ and $(m,n)$ are connected without ``using'' the events of $\mathcal{E} \backslash E_{k}$.
To do this efficiently, we further modify the Alt and Godau algorithm.
We require only a method to prevent events in $\mathcal{E} \backslash E_{k}$ from being used.
After $L^R_{i,j}$ is computed, we check whether the event $e$ (if any) that ends at the left boundary of cell $(i,j)$ is part of $\mathcal{E} \backslash E_{k}$ and necessary to obtain $L^R_{i,j}$.
If this is the case, we replace $L^R_{i,j}$ with an empty interval.
Event $e$ is necessary if and only if $L^R_{i,j}$ is a singleton.
To obtain an algorithm that is numerically more stable, we introduce entry points.
The \emph{entry point} of the left boundary of cell $(i,j)$ is the maximal $i' < i$ such that $B^R_{i',j}$ is non-empty.
These values are easily computed during the decision algorithm.
Assume the passage corresponding to event $e$ starts on the left boundary of cell $(i_\text{s}, j)$.
Event $e$ is necessary to obtain $L_{i,j}^R$ if and only if $i' < i_\text{s}$.
Therefore, we use the entry point instead of checking whether $L^R_{i,j}$ is a singleton.
This process is analogous for horizontal boundaries of cells.

Earlier we assumed that each event in $\mathcal{E}$ ends at a different cell boundary.
If events end at the same boundary, then these occur in the same row (or column) and it suffices to consider only the event that starts at the rightmost column (or highest row).
This justifies the assumption and ensures that $\mathcal{E}$ contains $O(m n)$ events.
Thus computing $e_\text{r}$ (Algorithm~\ref{alg:findmatching}, line~\ref{algline:find_er}) takes $O(m  n  \log(m  n))$ time, which is equal to the time needed to compute the \fd. Each recursion step splits the problem into two smaller problems, and the recursion ends when $mn \leq 1$. This results in an additional factor $m+n$. Thus the overall running time is $O((m+n) m  n  \log(m  n))$.

\begin{figure}[b]
  \centering
  \includegraphics{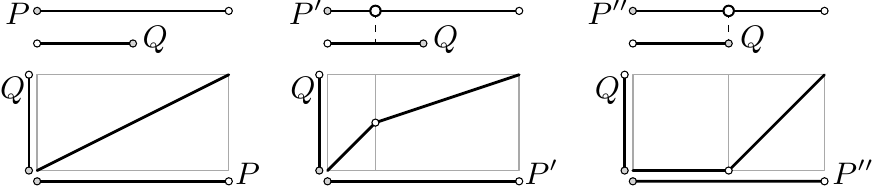}
  \caption{Different sampling may result in different matchings.}
  \label{fig:vertexdependency}
\end{figure}
\paragraph{Sampling and further restrictions}
Two curves may still have many locally correct \fm s: the algorithm computes just one of these.
However, introducing extra vertices may alter the result, even if these vertices do not modify the shape (see Fig.~\ref{fig:vertexdependency}).
This implies that the algorithm depends not only on the shape of the curves, but also on the sampling.
Increasing the sampling further and further seems to result in a matching that decreases the matched distance as much as possible within a cell.
However, since cells are rectangles, there is a slight preference for taking longer diagonal paths.
Based on this idea, we are currently investigating ``locally optimal'' \fm s.
The idea is to restrict to the locally correct \fm\ that decreases the matched distance as quickly as possible.

We also considered restricting to the ``shortest'' locally correct \fm, where ``short'' refers to the length of the path in the free space diagram.
However, Fig.~\ref{fig:lengthcounter} shows that such a restriction does not necessarily improve the quality of the matching.

\begin{figure}[t]
  \centering
  \includegraphics{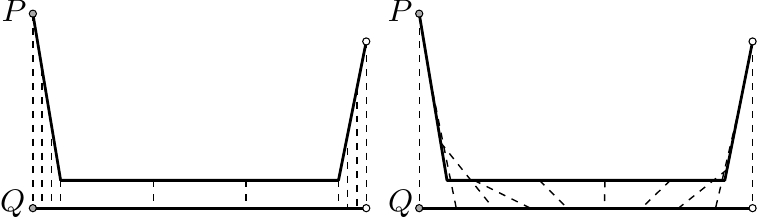}
  \shrinkfiguregap
  \caption{Two locally correct Fr\'echet matchings for $P$ and $Q$. Right: shortest matching.}
  \label{fig:lengthcounter}
\end{figure}

\section{Locally correct discrete \fm s}\label{sec:discrete}

Here we study the discrete variant of \fm s.
For the discrete \fd, only the vertices of curves are matched.
The discrete \fd\ can be computed in $O(m \cdot n)$ time via dynamic programming~\cite{em-cdfd-94}.
Here, we show how to also compute a locally correct discrete \fm\ in $O(m \cdot n)$ time.

\paragraph{Grids}
Since we are interested only in matching vertices of the curves, we can convert the problem to a grid problem.
Suppose we have two curves $P$ and $Q$ with $m$ and $n$ edges respectively.
These convert into a grid $G$ of non-negative values with $m+1$ columns and $n+1$ rows.
Every column corresponds to a vertex of $P$, every row to a vertex of $Q$.
Any node of the grid $G[i,j]$ corresponds to the pair of vertices $(p_i, q_j)$.
Its value is the distance between the vertices: $G[i,j] = |p_i - q_j|$.
Analogous to free space diagrams, we assume that $G[0,0]$ is the bottomleft node and $G[m,n]$ the topright node.

\paragraph{Matchings}
A monotonous path $\pi$ is a sequence of grid nodes $\pi(1), \ldots, \pi(k)$ such that every node $\pi(i)$ ($1 < i \leq k$) is the above, right, or above/right diagonal neighbor of $\pi(i-1)$.
In the remainder of this section a path refers to a monotonous path unless indicated otherwise.
A monotonous discrete matching of the curves corresponds to a path $\pi$ such that $\pi(1) = G[0,0]$ and $\pi(k) = G[m,n]$.
We call a path $\pi$ locally correct if for all $1 \leq t_1 \leq t_2 \leq k$, $\max_{t_1 \leq t \leq t_2} \pi(t) = \min_{\pi'} \max_{1 \leq t \leq k'} \pi'(t)$, where $\pi'$ ranges over all paths starting at $\pi'(1) = \pi(t_1)$ and ending at $\pi'(k') = \pi(t_2)$.

\begin{algorithm}[t]
  \caption{$\texttt{ComputeDiscreteLCFM}(P, Q)$}\label{alg:finddiscrete}
  \begin{algorithmic}[1]

  \REQUIRE $P$ and $Q$ are curves with $m$ and $n$ edges
  \ENSURE A locally correct discrete \fm\ for $P$ and $Q$

  \STATE Construct grid $G$ for $P$ and $Q$
  \STATE Let $T$ be a tree consisting only of the root $G[0,0]$
  \FOR{$i \leftarrow 1$ \textbf{to} $m$}
    \STATE Add $G[i,0]$ to $T$
  \ENDFOR
  \FOR{$j \leftarrow 1$ \textbf{to} $n$}
    \STATE Add $G[0,j]$ to $T$
  \ENDFOR
  \FOR{$i \leftarrow 1$ \textbf{to} $m$}
    \FOR{$j \leftarrow 1$ \textbf{to} $n$}
        \STATE $\texttt{AddToTree}(T, G, i, j)$
    \ENDFOR
  \ENDFOR
  \RETURN path in $T$ between $G[0,0]$ and $G[m,n]$
 \end{algorithmic}
\end{algorithm}

\paragraph{Algorithm}
The algorithm needs to compute a locally correct path between $G[0,0]$ and $G[m,n]$ in a grid $G$ of non-negative values.
To this end, the algorithm incrementally constructs a tree $T$ on the grid such that each path in $T$ is locally correct.
The algorithm is summarized by Algorithm~\ref{alg:finddiscrete}.
We define a \emph{growth node} as a node of $T$ that has a neighbor in the grid that is not yet part of $T$: a new branch may sprout from such a node.
The growth nodes form a sequence of horizontally or vertically neighboring nodes.
A \emph{living node} is a node of $T$ that is not a growth node but is an ancestor of a growth node.
A \emph{dead node} is a node of $T$ that is neither a living nor a growth node, that is, it has no descendant that is a growth node.
Every pair of nodes in this tree has a \emph{nearest common ancestor} (NCA).
When we have to decide what parent to use for a new node in the tree, we look at the maximum value on the path in the tree between the parents and their NCA (excluding the value of the latter).
A \emph{face} of the tree is the area enclosed by the segment between two horizontally or vertically neighboring growth nodes (without one being the parent of another) and the paths to their NCA.
The unique \emph{sink} of a face is the node of the grid that is in the lowest column and row of all nodes on the face.
Fig.~\ref{fig:sinksandshortcuts} (a-b) shows some examples of faces and their sinks.

\begin{figure}[b]
\centering
\includegraphics{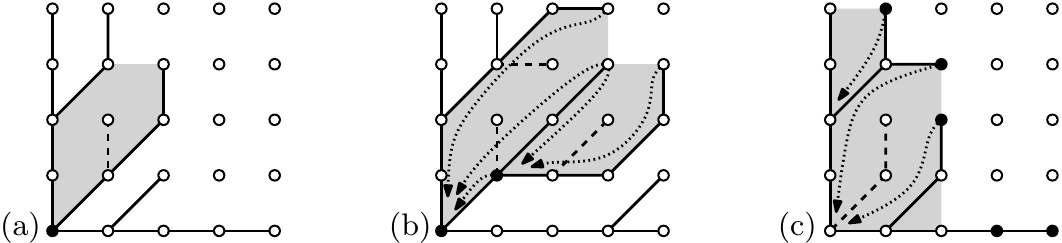}
\shrinkfiguregap
\caption{(a) Face of tree (gray area) with its unique sink (solid dot). A dashed line represents a dead path. (b) Two adjacent faces with some shortcuts indicated. (c) Tree with 3 faces. Solid dots indicate growth nodes with a growth node as parent. These nodes are incident to at most one face. All shortcuts of these nodes are indicated.}
\label{fig:sinksandshortcuts}
\end{figure}

\paragraph{Shortcuts}
To avoid repeatedly walking along the tree to compute maxima, we maintain up to two \emph{shortcuts} from every node in the tree.
The segment between the node and its parent is incident to up to two faces of the tree.
The node maintains shortcuts to the sink of these faces, associating the maximum value encountered on the path between the node and the sink (excluding the value of the sink).
Fig.~\ref{fig:sinksandshortcuts} (b) illustrates some shortcuts.
With these shortcuts, it is possible to determine the maximum up to the NCA of two (potentially diagonally) neighboring growth nodes in constant time.

Note that a node $g$ of the tree that has a growth node as parent is incident to at most one face (see Fig.~\ref{fig:sinksandshortcuts} (c)).
We need the ``other'' shortcut only when the parent of $g$ has a living parent.
Therefore, the value of this shortcut can be obtained in constant time by using the shortcut of the parent.
When the parent of $g$ is no longer a growth node, then $g$ obtains its own shortcut.

\paragraph{Extending the tree}
Algorithm~\ref{alg:extend} summarizes the steps required to extend the tree $T$ with a new node.
Node $G[i,j]$ has three \emph{candidate parents}, $G[i-1, j]$, $G[i-1, j-1]$, and $ G[i, j-1]$.
Each pair of these candidates has an NCA.
For the actual parent of $G[i,j]$, we select the candidate $c$ such that for any other candidate $c'$, the maximum value from $c$ to their NCA is at most the maximum value from $c'$ to their NCA---both excluding the NCA itself.
We must be consistent when breaking ties between candidate parents.
To this end, we use the preference order of $G[i-1, j] \succ G[i-1, j-1] \succ G[i, j-1]$.
Since paths in the tree cannot cross, this order is consistent between two paths at different stages of the algorithm.
Note that a preference order that prefers $G[i-1,j-1]$ over both other candidates or vice versa results in an incorrect algorithm.

\begin{algorithm}[t]
  \caption{$\texttt{AddToTree}(T, G, i, j)$}\label{alg:extend}
  \begin{algorithmic}[1]

  \REQUIRE $G$ is a grid of non-negative values; any path in tree $T$ is locally correct
  \ENSURE node $G[i,j]$ is added to $T$ and any path in $T$ is locally correct

  \STATE $parent(G[i, j]) \leftarrow$ candidate parent with lowest maximum value to NCA

  \IF{$G[i-1, j-1]$ is dead}
    \STATE Remove the dead path ending at $G[i-1, j-1]$ and extend shortcuts
  \ENDIF

  \STATE Make shortcuts for $G[i-1, j]$, $G[i, j-1]$, and $G[i, j]$ where necessary
 \end{algorithmic}
\end{algorithm}

\begin{figure}[b]
\centering
\includegraphics{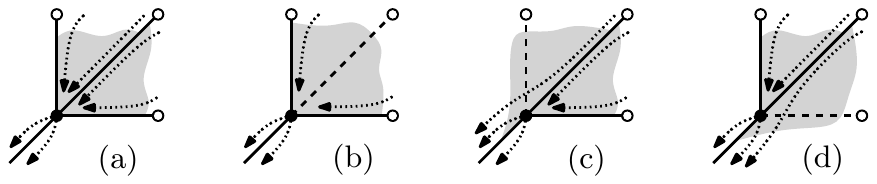}
\shrinkfiguregap
\caption{(a) Each sink has up to four sets of shortcuts. (b-d) Removing a dead path (dashed) extends at most one set of shortcuts.}
\label{fig:fourshortcutsets}
\end{figure}

When a dead path is removed from the tree, adjacent faces merge and a sink may change.
Hence, shortcuts have to be extended to point toward the new sink.
Fig.~\ref{fig:fourshortcutsets} illustrates the incoming shortcuts at a sink and the effect of removing a dead path on the incoming shortcuts.
Note that the algorithm does not need to remove dead paths that end in the highest row or rightmost column.

Finally, $G[i-1, j]$, $G[i, j-1]$, and $G[i, j]$ receive shortcuts where necessary. $G[i-1, j]$ or $G[i, j-1]$ needs a shortcut only if its parent is $G[i-1, j-1]$.
$G[i, j]$ needs two shortcuts if $G[i-1,j-1]$ is its parent, only one shortcut otherwise.

\paragraph{Correctness}
To prove correctness of Algorithm~\ref{alg:finddiscrete}, we require a stronger version of local correctness.
A path $\pi$ is \emph{strongly locally correct} if for all paths $\pi'$ with the same endpoints $\max_{1 < t \leq k} \pi(t) \leq \max_{1 < t' \leq k'} \pi'(t')$ holds.
Note that the first node is excluded from the maximum.
Since $\max_{1 < t \leq k} \pi(t) \leq \max_{1 < t' \leq k'} \pi'(t')$ and $\pi(1) = \pi'(1)$ imply $\max_{1 \leq t \leq k} \pi(t) \leq \max_{1 \leq t' \leq k'} \pi'(t')$, a strongly locally correct path is also locally correct.
Lemma~\ref{lem:discretecorrect} implies the correctness of Algorithm~\ref{alg:finddiscrete}.

\begin{lemma}\label{lem:discretecorrect}
Algorithm~\ref{alg:finddiscrete} maintains the following invariant:
any path in $T$ is strongly locally correct.
\end{lemma}
\begin{proof}
To prove this lemma, we strengthen the invariant.

\textbf{Invariant.}
We are given a tree $T$ such that every path in $T$ is strongly locally correct.
In constructing $T$, any ties were broken using the preference order.

\textbf{Initialization.}
Tree $T$ is initialized such that it contains two types of paths:
either between grid nodes in the first column or in the first row.
In both cases there is only one path between the endpoints of the path.
Therefore, this path must be strongly locally correct.
Since every node has only one candidate parent, $T$ adheres to the preference order.

\textbf{Maintenance.}
The algorithm extends $T$ to $T'$ by including node $g = G[i,j]$.
This is done by connecting $g$ to one of its candidate parents ($G[i-1,j]$, $G[i-1,j-1]$, or $G[i,j-1]$), the one that has the lowest maximum value along its path to the NCA.
We must now prove that any path in $T'$ is strongly locally correct.
From the induction hypothesis, we conclude that only paths that end at $g$ could falsify this statement.

Suppose that such an invalidating path exists in $T'$, ending at $g$.
This path must use one of the candidate parents of $g$ as its before-last node.
We distinguish three cases on how this path is situated compared to $T'$.
The last case, however, needs two subcases to deal with candidate parents that have the same maximum value on the path to their NCA.
The four cases are illustrated in Fig.~\ref{fig:nonLCcasesAppendix}.

For each case, we consider the path $\pi_\textrm{i}$ between the first vertex and the parent of $g$ in the invalidating path (i.e. one of the three candidate parents).
Note that $\pi_\textrm{i}$ need not be disjoint of the paths in $T'$.
Slightly abusing notation, we also use a path $\pi'$ to denote its maximum value, excluding the first node, i.e. $\max_{1 < t \leq k'} \pi'(t)$.
We now show that for each case, the existence of the invalidating path contradicts the invariant on $T$.

\begin{figure}[b]
\centering
\includegraphics{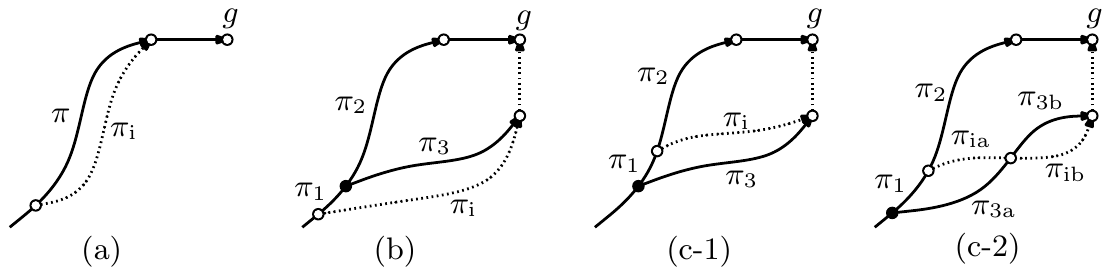}
\shrinkfiguregap
\caption{The four cases for the proof of Lemma~\ref{lem:discretecorrect}.}
\label{fig:nonLCcasesAppendix}
\end{figure}

\textit{Case (a).}
Path $\pi_\mathrm{i}$ ends at the parent of $g$ in $T'$.
Path $\pi$ is the path in $T'$ between the first and last vertex of $\pi_\mathrm{i}$.
Since $(\pi_\mathrm{i}, g)$ is the invalidating path, we know that $\max\{ \pi_\mathrm{i}, g \} < \max\{ \pi, g \}$ holds.
This implies that $\pi_\mathrm{i} < \pi$ holds.
In particular, this means that $\pi$, a path in $T$, is not strongly locally correct: a contradiction.

\textit{Case (b).}
Path $\pi_\mathrm{i}$ ends at a non-selected candidate parent of $g$.
Path $\pi_2$ ends at the parent of $g$ in $T'$ and path $\pi_3$ ends at the last vertex of $\pi_\mathrm{i}$.
Let $nca$ denote the NCA of the endpoints of $\pi_2$ and $\pi_3$.
The first vertex of $\pi_\mathrm{i}$ is $nca$ or one of its ancestors.
Both path $\pi_2$ and $\pi_3$ start at $nca$.
Let $\pi_1$ be the path from $\pi_\mathrm{i}(1)$ to $nca$.
Since the endpoint of $\pi_2$ was chosen as parent over the endpoint of $\pi_3$, we know that $\pi_2 \leq \pi_3$ holds.
Furthermore, since $(\pi_\mathrm{i}, g)$ is the invalidating path, we know that $\max\{ \pi_\mathrm{i}, g \} < \max\{ \pi_1, \pi_2, g \}$ holds.
These two inequalities imply $\max\{ \pi_\mathrm{i}, g \} < \max\{ \pi_1, \pi_3, g \}$ holds.
This in turn implies that $\pi_\mathrm{i} <  \max\{ \pi_1, \pi_3 \}$ must hold.
Since $(\pi_1, \pi_3)$ is a path in $T$ and the inequality implies that it is not strongly locally correct, we again have a contradiction.

\textit{Case (c).}
Path $\pi_\mathrm{i}$ ends at a non-selected candidate parent of $g$.
Path $\pi_2$ ends at the parent of $g$ in $T'$ and path $\pi_3$ ends at the last vertex of $\pi_\mathrm{i}$.
Let $nca$ denote the NCA of the endpoints of $\pi_2$ and $\pi_3$.
The first vertex of $\pi_\mathrm{i}$ is a descendant of $nca$.
Path $\pi_2$ starts at $\pi_\mathrm{i}(1)$ and $\pi_3$ starts at $nca$.
Let $\pi_1$ be the path from $nca$ to $\pi_\mathrm{i}(1)$.
In this case, we must explicitly consider the possibility of two paths having equal values.
Hence, we distinguish two subcases.

\textit{Case (c-1).}
In the first subcase, we assume that the endpoint of $\pi_2$ was chosen as parent since its maximum value is strictly lower: $\max\{ \pi_1, \pi_2 \} < \pi_3$ holds.
Since $(\pi_\mathrm{i}, g)$ is the invalidating path, we know that $\max\{ \pi_\mathrm{i}, g \} < \max\{ \pi_2, g \}$ holds.
Since $\pi_2 \leq \max\{ \pi_1, \pi_2 \}$ always holds, we obtain that $\max\{ \pi_\mathrm{i}, g \} < \max\{ \pi_3, g \}$ must hold.
This in turn implies that $\pi_\mathrm{i} <  \pi_3$ holds.
Similarly, since $\pi_1 \leq \max\{ \pi_1, \pi_2 \}$, we know that $\pi_1 < \pi_3$ must hold.
Combining these last two inequalities yields $\max\{ \pi_1, \pi_\mathrm{i} \} < \pi_3$.
Since $\pi_3$ is a path in $T$ and the inequality implies that it is not strongly locally correct, we again have a contradiction.
(Note that with $\max\{ \pi_1, \pi_2 \} \leq \pi_3$, we can at best derive $\max\{ \pi_1, \pi_\mathrm{i} \} \leq \pi_3$ which is not strong enough to contradict the invariant on $T$.)

\textit{Case (c-2).}
In the second subcase, we assume that the endpoint of $\pi_2$ was chosen as parent based on the preference order: the maximum values are equal, thus $\max\{ \pi_1, \pi_2 \} = \pi_3$ holds.
We now subdivide $\pi_\mathrm{i}$ into two parts, $\pi_\mathrm{ia}$ and $\pi_\mathrm{ib}$.
$\pi_\mathrm{ia}$ runs from the first vertex of $\pi_\mathrm{i}$ up to the first vertex along $\pi_\mathrm{i}$ that is an ancestor in $T'$ of candidate parent of $g$ that is used by $\pi_\mathrm{i}$.
At the same point, we also split path $\pi_3$ into $\pi_\mathrm{3a}$ and $\pi_\mathrm{3b}$.
We now obtain two more cases, $\pi_\mathrm{3a} < \max\{ \pi_1, \pi_\mathrm{ia}\}$ and $\pi_\mathrm{3a} \geq \max\{ \pi_1, \pi_\mathrm{ia}\}$.
In the former case, we obtain that $\max\{\pi_\mathrm{3a}, \pi_\mathrm{ib}\} < \max\{ \pi_1, \pi_2\}$ holds and thus $(\pi_\mathrm{3a}, \pi_\mathrm{ib}, g)$ is also an invalidating path.
Since this path starts at the NCA of $\pi_2$ and $\pi_3$, this is already covered by case (b).
In the latter case, we have that either path $\pi_\mathrm{3a}$---which is in $T$---is not strongly locally correct (contradicting the induction hypothesis) or there is equality between the two paths $(\pi_1, \pi_\mathrm{ia})$ and $\pi_\mathrm{3a}$.
In case of equality, we observe that $(\pi_1, \pi_\mathrm{ia})$ and $\pi_\mathrm{3a}$ arrive at their endpoint in the same order as $\pi_2$ and $\pi_3$ arrive at $g$.
Thus $T$ does not adhere to the preference order to break ties.
This contradicts the invariant.
\end{proof}

\paragraph{Execution time}
When a dead path $\pi_\mathrm{d}$ is removed, we may need to extend a list of incoming shortcuts at $\pi_\mathrm{d}(1)$, the node that remains in $T$.
Let $k$ denote the number of nodes in $\pi_\mathrm{d}$.
The lemma below relates the number of extended shortcuts to the size of $\pi_\mathrm{d}$.
The main observation is that the path requiring extensions starts at $\pi_\mathrm{d}(1)$ and ends at either $G[i-1,j]$ or $G[i,j-1]$, since $G[i,j]$ has not yet received any shortcuts.

\begin{lemma}\label{lem:deadpath}
A dead path $\pi_\mathrm{d}$ with $k$ nodes can result in at most $2 \cdot k - 1$ extensions.
\end{lemma}
\begin{proof}%
\begin{figure}[t]
\centering
\includegraphics{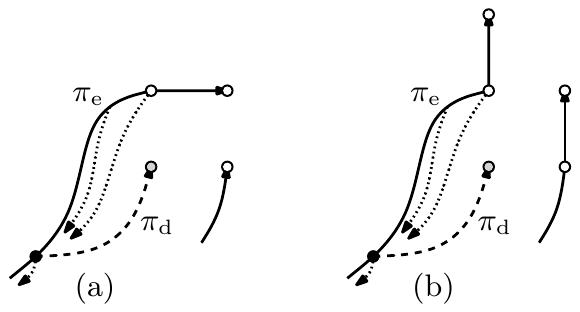}
\shrinkfiguregap
\caption{(a) A dead path $\pi_\mathrm{d}$ and the corresponding path requiring extensions $\pi_\mathrm{e}$.\newline(b) Endpoint of $\pi_\mathrm{e}$ has outdegree 1. None of its descendant has a shortcut to $\pi_\mathrm{e}(1)$.}
\label{fig:DeadPath}
\end{figure}
Since $\pi_\mathrm{d}$ is a path with $k$ nodes, it spans at most $k$ columns and $k$ rows.
When a dead path is removed, its endpoint is $G[i-1,j-1]$.
Let $\pi_\mathrm{e}$ denote the path of $T$ that requires extensions.
We know that both paths start at the same node: $\pi_\mathrm{d}(1) = \pi_\mathrm{e}(1)$.
The endpoint of $\pi_\mathrm{e}$ is at either $G[i-1,j]$ or $G[i,j-1]$, since $G[i,j]$ has not yet received shortcuts when the dead path is removed.
Also, note that if the endpoint of $\pi_\mathrm{e}$ is not the parent of $G[i,j]$ and has outdegree higher than 0, then it is a growth node and its descendants are also growth nodes.
Hence, these descendants have parent that is a growth node and thus have shortcuts that need to be extended.
Fig.~\ref{fig:DeadPath} illustrates these situations.
Hence, we know that $\pi_\mathrm{e}$ spans either $k + 1$ columns and $k$ rows or vice versa.
Therefore, the maximum number of nodes in $\pi_\mathrm{e}$ is $2 \cdot k$, since it must be monotonous.
Since $\pi_\mathrm{e}(1)$ does not have a shortcut to itself, there are at most $2 \cdot k - 1$ incoming shortcuts from $\pi_\mathrm{e}$ at $\pi_\mathrm{d}(1)$.
\end{proof}
Hence, we can charge every extension to one of the $k-1$ dead nodes (all but $\pi_\mathrm{d}(1)$).
A node gets at most 3 charges, since it is a (non-first) node of a dead path at most once.
Because an extension can be done in constant time, the execution time of the algorithm is $O(m n)$.
Note that shortcuts that originate from a living node with outdegree 1 could be removed instead of extended.
We summarize the findings of this section in the following theorem.

\begin{theorem}
Algorithm~\ref{alg:finddiscrete} computes a locally correct discrete \fm\ of two polygonal curves $P$ and $Q$ with $m$ and $n$ edges in $O(m n)$ time.
\end{theorem}

\section{Conclusion}\label{sec:conclusion}

We set out to find ``good'' matchings between two curves. 
To this end we introduced the local correctness criterion for \fm s.
We have proven that there always exists at least one locally correct \fm\ between any two polygonal curves.
This proof resulted in an $O(N^3 \log N)$ algorithm, where $N$ is the total number of edges in the two curves.
Furthermore, we considered computing a locally correct matching using the discrete \fd.
By maintaining a tree with shortcuts to encode locally correct partial matchings, we have shown how to compute such a matching in $O(N^2)$ time.

\paragraph{Future work}
Computing a locally correct discrete \fm\ takes $O(N^2)$ time, just like the dynamic program to compute only the discrete \fd.
However, computing a locally correct continuous \fm\ takes $O(N^3 \log N)$ time, a linear factor more than computing the \fd.
An interesting question is whether this gap in computation can be reduced as well.

Furthermore, it would be interesting to investigate the benefit of local correctness for other matching-based similarity measures, such as the geodesic width~\cite{Efrat2002}.

\bibliographystyle{abbrv}
\bibliography{references}

\end{document}